\documentclass[12pt]{article}

\usepackage{customtemplate}

\newcommand{\F}{\mathbb{F}}

\newcommand{\X}{\mathbf{X}}

\def\Fq{{\mathbb F}_q}

\def\d{{\delta}}

\def\X{{\textbf{X}}}
\def\t{{\times}}
\def\AA{{\mathbb A}}

\def\PP{{\mathbb P}}
\def\clm{{C^\AA(\ell,m)}}
\def\clmD{{C^\AA_D(\ell,m)}}
\def\l{{\ell}}

\DeclareMathOperator{\Supp}{Supp}
\DeclareMathOperator{\ev}{Ev}
\DeclareMathOperator{\aut}{Aut}
\DeclareMathOperator{\GL}{GL}

\theoremstyle{definition}

\numberwithin{theorem}{section}

\title{Majority Logic Decoding of Affine Grassmann Codes Over Nonbinary Fields}
\author{ Fernando Pi\~nero Gonz\'alez, Prasant Singh, and Rohit Yadav}
\date{}
\begin{document}      
\maketitle
\begin{abstract}
In this article, we consider the decoding problem of affine Grassmann codes over nonbinary fields. We use matrices of different ranks to construct a large set consisting of parity checks of affine Grassmann codes, which are orthogonal with respect to a fixed coordinate. By leveraging the automorphism groups of these codes, we generate a set of orthogonal parity checks for each coordinate. Using these parity checks, we perform majority logic decoding to correct a large number of errors in affine Grassmann codes. The order of error correction capability and the complexity of this decoder for affine Grassmann codes are the same as those of the majority logic decoder for Grassmann codes proposed in \cite{BPP2021}. 
\end{abstract}

\section{Introduction}

Let $q$ be the power of a prime, and let $\Fq $ be the finite field with $q$ elements. Using the language of projective systems \cite{TVN2007}, one can associate a linear code corresponding to subsets of $\Fq$-rational points of algebraic varieties. Affine Grassmann codes are examples of such a class of codes. These codes are obtained by taking the projective system of the set of \( \Fq \)-rational points of an affine open subset of the Grassmannian. They can be considered a generalization of the first-order Reed-Muller codes. For example, let $\l$ and $\l^\prime$ be positive integers such that $\ell \leq \ell^\prime $. Set $ m = \ell + \ell^\prime $ and $ \delta = \ell \ell^\prime $. Let $ \AA^\delta(\Fq) $ denote the set of all $ \Fq $-rational points of an affine open cell of the Grassmannian $ G_{\ell, m} $, which is defined by the non-vanishing of some fixed coordinate. The code corresponding to this projective system is known as the affine Grassmann code corresponding to the set $ \AA^\delta(\Fq) $ and is denoted by $ \clm $. The study of affine Grassmann codes was initiated by Beelen-Ghorpade-H{\o}holdt \cite{BGH2010}, and in this article, the basic parameters of these codes were determined. To be precise, it was proved that the affine Grassmann code $ \clm $ is an $ [n, k, d]_q $ code where

 \begin{equation}
     \label{eq: parameter}
     n=q^\d, \; k= \binom{m}{\l}, \text{ and }d=q^{\d-\ell^2}\prod_{i=0}^{\l-1}(q^\l-q^i).
 \end{equation}

In addition, in this article, the number of codewords of minimum weight was also determined, and it was shown that the automorphism group of $\clm$ is quite large. In a subsequent work \cite{BGT2012}, the authors studied the dual Grassmann code $\clm^\perp$ in detail and determined that this code is the evaluation of certain functions on the set of all matrices of size $\ell\times \ell^\prime$. Further, it was shown that the minimum distance of the dual affine Grassmann code $\clm^\perp$ is given by 
\begin{equation}
    \label{eq: mindistdual}
    d=\begin{cases}
        3, \text{ if }q\geq 3\\
        4 \text{ otherwise}.
    \end{cases}
\end{equation}

Since their introduction, several properties of affine Grassmann codes have been studied. For example, Ghorpade and Kaipa \cite{GK2013} explicitly determined the automorphism groups of $\clm$. Some initial and terminal generalized Hamming weights of $\clm$ were determined by Datta and Ghorpade \cite{DG2015}. The weight spectrum of $\clm$ in the case $\ell=2$ was determined in \cite{PS2019}. In other words, over the last 15 years, different groups of mathematicians have studied several interesting properties of affine Grassmann codes. However, the decoding problem for these codes has not been explored.

As we mentioned earlier, affine Grassmann codes are codes obtained from the set of $\Fq$-rational points of the Grassmannian. One can easily see that these two classes of codes have the same order of minimum distances and dimensions. In addition, several other properties of these two classes of codes are very similar. For example, the weight spectrum of the Grassmann code $C(2, m)$, and the affine Grassmann code $C^{\mathbb{A}}(2, m)$ can be determined by the rank of the matrix corresponding to codewords of these codes \cite{Nogin1996, PS2019}. Similarly, the automorphism groups of these two classes of codes are also of the same order \cite{GK2013}. Recently, using the point-line incidence geometry of the Grassmannian, a large number of orthogonal parity checks for the Grassmann code $C(\ell, m)$ were constructed. As a consequence, the majority voting decoder for Grassmann codes was proposed \cite{BPP2021}. The proposed decoder has a complexity of order 2 in the length of the Grassmann code and can asymptotically correct up to $\lfloor d/2^{\ell+1}\rfloor$ errors for the code $C(\ell, m)$. Further, a majority voting decoder was proposed for certain Schubert codes \cite{S2022}, which are also codes obtained from the projective system of $\Fq$-rational points of Schubert subvarieties of the Grassmannian. Therefore, it is natural to explore the decoding problem for affine Grassmann codes through the prism of majority voting.

This article focuses only on affine Grassmann codes $\clm$ over nonbinary fields. We aim to construct a set of parity checks for the affine Grassmann code $\clm$ that can be used to perform majority voting decoding. Since the automorphism group $\aut(\clm)$  acts transitively on the coordinates of codewords of $\clm$, it is sufficient to construct parity checks that are orthogonal to some fixed coordinate. Our idea is to first construct parity checks whose support consists of matrices of the same rank and are orthogonal. We will show that if we choose the right set of matrices of a fixed rank, then puncturing the affine Grassmann code outside these matrices results in a parity check of $\clm$ with support lying in the chosen set of matrices. These parity checks will serve as our candidates for orthogonal parity checks.

\section{Preliminaries}
In this section we recall some known results for affine Grassmann codes and the idea of majority decoding for linear codes.
The section is divided into two parts. In the first part, we recall the construction of affine Grassmann codes and also some results that are going to be useful for us. In the second part, we discuss the notion of orthogonal parity checks for a code. The section ends with one important results about using orthogonal parity checks to correct certain errors for the underlying code.
\subsection{ Affine Grassmann Codes}

As we have fixed, let $\F_q$ be the finite field with $q$ elements, where  $q$ is a power of some fixed prime. Fix positive integers $\l,\l^\prime$  with $\l \leq \l^\prime$; we set $m=\l + \l^\prime$ and $\delta = \l \l^\prime.$  Let $\AA^\d(\Fq)$ be the set of all matrices of size $\ell\times \ell^\prime$ over $\Fq$. This set can be identified as the set of $\Fq$-rational points of an open subset of the Grassmannian $G_{\ell, m}$ of all $\ell$-planes of an $m$-dimensional vector space over $\Fq$ (see \cite{BGH2010}*{Section VII} for further details).

Affine Grassmann codes are codes obtained from the projective system  $\AA^\d(\Fq)  \subseteq \mathbb{P}(\Fq)^{{m\choose\ell}-1}$. Briefly, the construction of affine Grassmann codes as in \cite{BGH2010} is as follows: Let  $\mathcal{F}(\l,m)$ be the $\F_q$-linear space generated by all minors of the generic $\ell \times \ell^\prime$ matrix $\X = (X_{ij})$, in $\delta$ variables $X_{i,j}$.  Note that, according to convention, the  $0 \times 0$ minor of $\X$ is the constant function $1$. In fact $\mathcal{F}(\l,m)$ is a $\binom{m}{\l}$-dimensional subspace of $\Fq[\X]$, where $\Fq[\X]$ is a polynomial ring in $\d$ indeterminates $X_{ij}, ( 1 \leq i \leq \ell, 1 \leq j \leq \l^\prime)$.  Fix an enumeration $P_1, P_2,...,P_{q^\delta}$ of $\AA^\d(\Fq)$. Consider the evaluation map

\begin{align*}
\mathrm{\ev}: \mathcal{F}(\ell,m) &\longrightarrow \Fq^{q^\delta} \\
f &\mapsto (f(P_1), f(P_2), \dots, f(P_{q^\delta})).
\end{align*}

The evaluation map is an injective $\mathbb{F}_q$-linear map and the image of $\ev$ is a linear code. This code is called the affine Grassmann code and is denoted by $\clm$.  The length $n$, the dimension $k$, and the minimum distance $d$ of $\clm$  were determined by Beelen-Ghorpade-Hh{\o}ldt in \cite{BGH2010} and are given by equation \eqref{eq: parameter}. In a subsequent article, Beelen-Ghorpade-H{\o}holdt \cite{BGT2012} also studied the dual of affine Grassmann codes. In fact, they described the code $\clm^\perp$ as an evaluation of certain functions on points of $\AA^\d(\Fq)$. They showed that the minimum distance of the code $\clm^\perp$ is given by equation \eqref{eq: mindistdual}. Moreover, in this article, they explicitly determined a subgroup of the automorphism group $\aut(\clm)$ of $\clm$. The complete automorphism group $\aut(\clm)$ was later determined by Ghorpade and Kaipa in \cite{GK2013}. For example, for any $B\in \GL_\l(\Fq)$ and $A \in \GL_{\l'}(\Fq)$, and $u\in \AA^\d(\Fq)$, consider the following affine transformation:

\begin{align*}
  \psi_{u,A,B}:\AA^\d(\Fq)  \longrightarrow & \AA^\d(\Fq) \\
    P \mapsto & BPA^{-1}+u.        
\end{align*}

 It is clear that the transformation $\psi_{u,A,B}$ gives a bijection of $\AA^\d(\Fq):=\{P_1,\ldots,P_{q^\d}\}$ onto itself, and hence there is a unique permutation $\sigma$ of $\{1,\ldots,q^\d \}$ such that $$(\psi_{u,A,B}(P_1),\ldots, \psi_{u,A,B}(P_{q^\d}))=(P_{\sigma(1)},\ldots,P_{\sigma(q^\d)}).$$
We shall denote this permutation $\sigma$ by $\sigma_{u,A,B}$ and for any $c=(c_1,\ldots,c_{q^\d}),$ we will often write $\sigma_{u,A,B}(c)$ for the $n$-tuple $(c_{\sigma(1)},\ldots,c_{\sigma(q^\d)}).$ The following result shows that these permutations preserve the structure of the code $\clm$.
 
\begin{lemma} \cite{BGT2012}*{Lemma $7$}\label{automorphism}
\label{lemma: Autgrp}
    Let $B\in \GL_\l(\Fq)$, $A \in \GL_{\l^\prime}(\Fq)$, and $u\in \AA^\d(\Fq)$. Then $\sigma_{u,A,B} \in  \aut(\clm)$. 
\end{lemma}

 Note that the group of affine transformations
 $$
 \{\sigma_{u,A,B} \mid u \in \AA^\delta(\Fq), A \in \GL_{\l^\prime}(\Fq), B \in \GL_{\l}(\Fq)\}
 $$
 acts transitively on the coordinates of $\clm$. This subgroup plays a crucial role in the proof of the main result in the next section.

\subsection{Majority logic decoding}
Majority logic decoding is one of the most notable decoding algorithms for certain classes of linear codes. It is typically applicable to codes that admit a large matrix consisting of orthogonal parity checks. For example, Reed used iterative majority logic decoding for binary Reed-Muller codes \cite{MS1977}*{Ch 13.7}. More recently, a majority logic decoder was proposed for Grassmann codes and certain classes of Schubert codes \cites{BPP2021, S2022}.  Thus, before proceeding further, and for the sake of completeness, we recall the definition of orthogonal parity checks and how they can be used for error correction. For more details on orthogonal parity checks for binary codes, we refer to \cite{MS1977}*{Ch 13.7}; for $q $-ary codes, we refer to \cite{M1963}*{Ch 1}.

\begin{definition}
     Let $C$ be an $[n, k]$ code. A set $\mathcal J$ of $J$ parity checks of $C$ is said to be orthogonal on the $i^{\it th}$ coordinate if the $J\times n$ matrix $H$  formed by these parity checks has: 
  	\begin{enumerate}
  		\item Each entry in the $i^{\it th}$ column of $H$ is $1$,
  		\item Any other column with at most one non-zero entry.
  	\end{enumerate}
\end{definition}

 The following theorem guarantees that if one has a large number of orthogonal parity checks for each coordinate, then many errors in the code can be corrected.

 \begin{theorem} \cite{M1963}*{Ch. 1, Thm. 1}
 \label{Majority_Thm}
 If for each $1 \le i \le n$, there exists a set $\mathcal J$ of $J$ orthogonal parity checks on the $i^{\it th}$ coordinate. Then, the corresponding majority logic decoder corrects up to $\lfloor J/2 \rfloor$ errors. 
  \end{theorem} 
A detailed proof of this theorem can be found in \cite{BPP2021}*{Thm. 4.2}, while the theorem itself originally appeared in \cite{M1963}*{Ch. 1, Thm. 1}. It is well known that the support of a minimum weight codeword of the binary Reed–Muller code $R(r, m)$ forms an $(m-r)$-dimensional flat of the  $m$-dimensional affine space. In $1954$, Reed \cite{Reed1954} exploited this property of Reed-Muller code and proposed an iterative decoder for binary Reed-Muller code, and it was shown that the proposed decoder could correct optimal errors for Reed-Muller code.

Similarly, it was shown \cite{BP2016} that the support of the minimum weight parity checks of Grassmann code $C(\ell, m)$ also has nice geometry. In fact, the support of the minimum weight parity checks of the Grassmann codes $C(\ell, m)$ lies on lines in the Grassmannian $G_{\l, m}$ and choosing any three points on such a line gives a minimum weight parity check for the Grassmann code. Recently, Beelen and Singh \cite{BPP2021} used these properties of Grassmann codes to construct unique geodesics between two points of Grassmannian, and these geodesics were used to obtain a huge number of orthogonal parity checks for each coordinate. As a consequence, majority voting decoding could be performed for Grassmann codes. It was shown that the proposed decoder could correct up to $\lfloor (d-1)/2^{\ell+1}\rfloor$ error for Grassmann code $C(\l, m)$ where $d$ is the minimum distance of the code. Additionally, Schubert codes are codes obtained by puncturing Grassmann code on points outside Schubert varieties. In a subsequent work, the second-named author showed \cite{S2022} that  a majority voting decoding can be proposed for certain Schubert codes, and the error-correcting capability of this decoder is of the same order as for Grassmann codes.

  The affine Grassmann codes are also obtained by puncturing the Grassmann code. In the nonbinary fields case, the support of the minimum weight parity checks of $\clm$ also lies on certain lines in the affine space $\AA^\d(\Fq)$. Therefore, exploring the decoding problem for affine Grassmann code for nonbinary cases using majority voting is natural.

\section{Majority logic decoding of affine Grassmann codes}


This is the main section of this article. In this section, we will construct orthogonal parity checks for nonbinary affine Grassmann codes. But before that we fix some notations for this section. We assume that $q \geq 3$ and as earlier $\mathbf{X}$ is the $\l \t \l'$ matrix in $\d$ variables $X_{ij}$; 
For subsets $I, J \subseteq [\l]$, with $|I|=|J|$, we denote by  $X_{I,J}$,  the minor of $\X$ obtained by taking the submatrix indexed by rows of $I$ and columns of $J$ and as standard convention, the $0 \times 0$ minor of $\X$ is the constant function $1$. Fix an ordering $P_1, P_2,...,P_{q^\delta}$ of all points of  $\AA^\d(\Fq)$. Recall that there is a one to one correspond between a codewords  $c = (c_{P_1}, c_{P_2}, \ldots, c_{P_{q^\d}}) \in \clm$ and functions $f\in \mathcal{F}(\ell,m)$ therefore we may represent a codeword as $c_f$. For a codeword $c_f\in\clm$ the support is defined as: 

$$
\Supp(c_f) := \{P \in \AA^\delta(\F_q): c_{P}\neq 0 \}.
$$

For any subset  $D:=\{ P_1, P_2,\ldots, P_k\}\subseteq\AA^\d(\Fq) $ the restriction of the affine Grassmann code $\clm$ onto the set $D$ the linear code obtained by taking the image  of restriction of the evaluation map $\ev_D$, where $\ev_D$ is defined as:

\begin{align*}
\ev_D : \mathcal{F}(\ell,m) &\longrightarrow \Fq^{k} \\f &\mapsto (f(P_1), f(P_2),\ldots, f(P_k)).
\end{align*}  

In otherwords, the restriction of the code $\clm$ on the set $D$ is code obtained by puncturing the affine Grassmann code $\clm$ on the set $\AA^\d(\Fq)\setminus D$. We denote this code by $C_D^{\mathbb{A}}(\ell, m)$

To continue any further, we need to fix some notations. Let  $V_r$ and  $W_r$ be  $r$-dimensional linear subspaces of $\Fq^\l$ and of $\Fq^{\l'}$ respectively. Fix a  basis $B=\{x_1, x_2, \ldots, x_r\}$ of $V_r$ and  a basis $B'=\{y_1, y_2, \ldots, y_r\}$ of $W_r$. Note that we treat vectors $x_i$ and $y_j$ as row vectors. For any $r$-tuple $\mathcal{A}_r= (A_1,A_2,\ldots,A_r)$, where each $A_i$ is a subset of $\Fq\setminus \{0\}$  with $|A_i|=2$ for $i=1,2,\dots,r$ we introduce the notation:
$$
M(B,B';\mathcal{A}_r) := \left\{ \sum\limits_{i=1}^r a_i x_i^Ty_i: a_i \in A_i \right\}. 
$$

Clearly, $M(B,B';\mathcal{A}_r)$ is consisting of  $\l \t \l'$ matrices with entries from $\Fq$. If the bases of $V_r$ and $W_r$ are standard basis vectors, then we simply write 
$$M(\mathcal{A}_r) := \left\{ \sum\limits_{i=1}^r b_i E_{i,i}: b_i \in A_i \right\}, $$
where $E_{i,i}$ denotes the $\l \t \l'$ matrix with $1$ in the $(i,i)$-th position and $0$ elsewhere. Since each $A_i$ consists of nonzero elements, all matrices in the set $M(\mathcal{A}_r)$ are of rank $r$ whose $i^{\it th}$ diagonal entry is coming from the set $A_i$.

The following lemma is useful in constructing orthogonal parity checks for the code $\clm$ using the sets $M(\mathcal{A}_r)$.
\begin{lemma}
\label{lem:projection_dimension}
Let $1 \leq r \leq \l$. For $1 \leq i \leq r$, let $A_i = \{a_{i1}, a_{i2}\}$ be a set of two distinct, nonzero $\F_q$ elements and let $\mathcal{A}_r= (A_1,A_2,\ldots,A_r)$. The restriction of the affine Grassmann code $\clm$ onto the set $M(\mathcal{A}_r) \cup \{0\}$ is a code of length $2^r+1$ and dimension $2^r$.
\end{lemma}
\begin{proof}

Clearly, the set $D:= M(\mathcal{A}_r) \cup \{0\}$ consists of $0$ matrix together with diagonal matrices of rank $r$ whose $i^{\it th}$ diagonal entry is $a_{i1}$ or $a_{i2}$ for $1\leq i\leq r$ and rest of the entries are zero. Hence $|D|= 2^r+1$. Now if we consider the matrix 
$$
\X'=\sum\limits_{i=1}^r X_{ii} E_{i,i},
$$
then for any $I,\; J\subseteq [\ell]$ with $|I|=|J|$, we have 
$$
\ev_D (X_{I,J}) =\ev_D (\X'_{I,J}). 
$$
Note that 
$$
\X'_{I,J} =\begin{cases}
    \prod\limits_{i \in I} X_{ii},\; \text{ if }I,\; J\subset [r]\text{ and }I=J\\
    0, \text{ otherwise }.
\end{cases}
$$
The restriction of the code $\clm$ onto the set $D$ is the same code as the code obtained by taking the evaluation map on the vector space $\mathcal{F}_r(\ell,m)$ spanned by all minors of $\X^\prime$ size almost $r$. Note that the dimension of the space $\mathcal{F}_r(\ell,m)$ is $2^r$ and the evaluation map $\ev_D$ is injective on this set. The injectivity follows from the following argument:

Suppose $f\in \mathcal{F}_r(\ell,m)$ is a polynomial that vanishes on every point $P\in D$, then $f$ belongs to ideal $\langle f_1(X_{11}), f_2(X_{22}), \ldots, f_r(X_{rr}) \rangle$, where $f_i(X_{ii}) = (X_{ii}-a_{i1})(X_{ii}-a_{i2})$. Since the functions $\mathcal{F}_r(\ell,m)$  have at most degree $1$ on each variable, it follows from the degree comparison that $f=0$. Consequently, $\ev_D$ gives an injective map on the set $\mathcal{F}_r(\ell,m)$.

This proves that the restriction of the affine Grassmann code $\clm $ onto the set $D$ is linear code of length $2^r+1$ and of dimension $2^r$. This completes the proof of the lemma.
\end{proof}
Using the lemma above we can produce parity checks of the code $\clm$ which support contains the zero matrix and lies entirely in the set of course chosen matrix. To be precise, we have the following result.
\begin{proposition}
    \label{prop: support}
    For any  $1 \leq r \leq \l$ and $\mathcal{A}_r=(A_1,\dots, A_r)$  as earlier, there exist a  parity check $w\in\clm^\perp$ satisfying
    $$
    0\in \Supp (w)\subseteq M(\mathcal{A}_r) \cup \{0\}.
    $$
    \end{proposition}

\begin{proof} 
As earlier, set $D=M(\mathcal{A}_r)\cup\{0\}$. We have seen in Lemma \ref{lem:projection_dimension} that the restriction $\clmD$ of the affine Grassmann code $\clm$ onto the set $D$ is an $[2^r+1, 2^r]$-linear code. Therefore, its dual $\clmD^\perp$ is generated by a single parity check, say $w^\prime$. Note that, $0$ matrix lies in the support of this parity check as puncturing $\clmD$ at the $0$ matrix gives a $[2^r, 2^r]$ code and puncturing $w^\prime$ at $0$ matrix lies the dual of this code, which is the trivial code. This guarantees that $0\in\Supp(w^\prime)$. Also, by construction $\Supp(w^\prime)\subseteq M(\mathcal{A}_r) \cup \{0\}$. Now we extend the parity check $w^\prime$ to a parity check $w$ by simply giving $0$ weight at coordinates outside the set $D$, i.e. the $P^{\it th}$ coordinate of $w$ is defined as
$$
w_P=\begin{cases}
    w^\prime_P, \text{ if }P\in D,\\
    0,\; \text{ otherwise }.
    \end{cases}
$$
Clearly, $w\in\clm^\perp$ and $\Supp(w)=\Supp(w^\prime)$. Therefore,

 $$
    0\in \Supp (w)\subseteq M(\mathcal{A}_r) \cup \{0\}.
    $$
This completes the proof of the Proposition.
\end{proof}

Next, we find more parity checks for affine Grassmann codes $\clm$ whose supports contains the zero matrix and certain chosen matrices of rank $r$ such that all these parity check are orthogonal at the $0$ matrix. To prove this we need Lemma \ref{lemma: Autgrp}. As earlier, let $V_r$ and $W_r$ be $r$-dimensional linear subspaces of $\Fq^\ell$ and $\Fq^{\ell^\prime}$ respectively. Let $B=\{ x_1, x_2, \ldots, x_r\}$  and $B'=\{y_1, y_2, \ldots, y_r\}$ be an ordered bases of $V_r$ and  $W_r$ respectively.

\begin{corollary}
\label{cor: ParityCheks}
For any  $1 \leq r \leq \l$ let $\mathcal{A}_r=(A_1,\dots, A_r)$, $V_r$, and $W_r$  be as earlier. Then given any ordered bases  $B$ and $B'$ of $V_r$ and $W_r$ respectively, there exists a parity check $w\in\clm^\perp$ satisfying
$$
0\in\Supp(w)\subseteq M(B,B';\mathcal{A}_r)\cup\{0\}.
$$
\end{corollary}

\begin{proof} 
Let $\{ e_1, e_2, \ldots, e_\l\}$ be the standard basis of $\Fq^\ell$ and let $\{ e_1^\prime, e_2^\prime, \ldots, e_{\ell^\prime}^\prime\}$ be the standard basis of $\Fq^{\ell^\prime}$. Consider the given ordered bases $B=\{ x_1, x_2, \ldots, x_r\}$ of the subspace $V_r$ and $B'=\{y_1, y_2, \ldots, y_r\}$ of subspace $W_r$. Choose matrices $Q \in \GL_\l(\Fq)$ and $Q' \in \GL_{\l'}(\Fq)$ such that $Qx_i=e_i$  and  $y_i^TQ'=e_i^\prime$  for each  $1 \leq i \leq r$ where $y_i^T$ denotes the transpose of $y_i$. If $D=M(\mathcal{A}_r)\cup \{0\}$ and $E=M(B, B^\prime;\mathcal{A}_r)\cup \{0\}$
then Lemma \ref{automorphism}, the automorphism $\sigma_{0, Q', Q} \in \aut(\clm)$ maps $C^{\AA}_E(\ell, m)$  onto the code $\clmD$. Now, by Proposition \ref{prop: support}, there exists a parity check $w_1\in\clm^\perp$ of satisfying
$$
0\in\Supp(w_1)\subseteq M(\mathcal{A}_r)\cup\{0\}.
$$
Now $\sigma_{0, Q', Q}^{-1}(w_1) =w$ is  the desired parity check of $\clm$.
\end{proof}

Recall that our aim is to construct many orthogonal parity checks, for example, ones that are orthogonal on the zero matrix coordinate. In Corollary \ref{cor: ParityCheks}, we constructed a parity check whose support includes the zero matrix and some matrices from the set $ M(B, B'; \mathcal{A}_r) $, where $B$ and $B^\prime$ are ordered bases of $r$-dimensional subspaces of $\Fq^\ell $ and $\Fq^{\ell'} $, respectively, and $ \mathcal{A}_r $ is an $r$-tuple of subsets of $\Fq$, each containing two nonzero elements.

Our next goal is to identify conditions under which the sets $ M(B, B'; \mathcal{A}_r) $ are disjoint. In the following lemma, we determine a condition for the disjointness of the sets $ M(B, B'; \mathcal{A}_r) $ when $B$ and $ \mathcal{A}_r $ are fixed, but $ B' $ varies. Let us fix $1\leq r\leq \ell$ be fixed and $ \mathcal{A}_r=(A_1,\dots, A_r)$ be as earlier.

\begin{lemma}
\label{lemm:Disjoint}
Let $V_r$ and $W_r$ be $r$-dimensional subspaces of the vector spaces $\Fq^\ell$ and $\Fq^{\ell^\prime}$ respectively. Let $B=\{x_1, x_2, \ldots, x_r\}$ be a fixed ordered basis of $V_r$ and let $B_1=\{y_1, y_2, \ldots, y_r\}$ and $B_2=\{z_1, z_2, \ldots, z_r\}$ be two ordered basis of $W_r$. If at least one of the pairs $\{y_i, z_i\}$ is linearly independent, then the sets $M(B, B_1; \mathcal{A}_r)$ and $M(B, B_2; \mathcal{A}_r)$ are disjoint.

\end{lemma}

\begin{proof}
Without loss of generality, we may assume that the basis $x_i$ is the $i^{\it th}$ standard basis vector of $\Fq^\l$. Thus the $i ^{\it th}$ row of  matrix $\sum\limits_{i=1}^r a_i x_i^Ty_i $ is $a_iy_i$. Likewise the $i ^{\it th}$ row of  matrix $\sum\limits_{i=1}^r b_i x_i^Tz_i $ is $b_iz_i$. 
Now if a matrix $P\in M(B, B_1; \mathcal{A}_r)\cap M(B, B_2; \mathcal{A}_r)$, then 
$$
P= \sum\limits_{i=1}^r a_i x_i^Ty_i = \sum\limits_{i=1}^r b_i x_i^Tz_i  \text{ for some }a_i, b_i\in A_i.
$$
Comparing the $i ^{\it th}$ row of $P$, we get $a_i y_i = b_i z_i$. This true for all $1\leq i\leq \ell$. Therefore, $\frac{a_i}{b_i} y_i = z_i$ for all $1\leq i\leq \ell$ but this contradicts the fact that  at least one of the pairs $\{y_i, z_i\}$ is linearly independent
\end{proof}

\begin{lemma} \label{lem:Disjoint_sets}
Let $V_r$ and $W_r$ be as earlier and let $B=\{x_1, x_2, \ldots, x_r\}$ be an ordered basis of $V_r$ and let $B'=\{y_1, y_2, \ldots, y_r\}$ be an ordered basis of $W_r$. Let $\mathcal{A}_r=(A_1,\dots, A_r)$ and $\mathcal{B}_r=(B_1,\dots, B_r)$ be two $r$-tuples of subsetes $\Fq^{*}$ with $|A_i|=|B_i|=2$ for all $1\leq i\leq r$. Then $M(B,B';\mathcal{A}_r)\cap M(B,B';\mathcal{B}_r) =\emptyset$ if and only if $A_i \cap B_i = \emptyset$ for some $i$.

\end{lemma}
\begin{proof}
Let us write $A_i = \{a_{i1}, a_{i2}\}$ and $B_i = \{b_{i1}, b_{i2}\}$ for every $1\leq i\leq r$. Assume that 
$$
M\in M(B,B';\mathcal{A}_r)\cap M(B,B';\mathcal{B}_r).
$$
Then there exists $a_i\in A_i$ and $b_i\in B_i$ such that $M= \sum\limits_{i=1}^r a_i x_i^Ty_i = \sum\limits_{i=1}^r b_i x_i^Ty_i $. This implies that $\sum\limits_{i=1}^r (a_i-b_i) x_i^Ty_i=0$. But this is true if and only if $a_i=b_i$ for all $1\leq i\leq r$. Or equivalently, $A_i\cap B_i\neq \emptyset$ for every $1\leq i\leq r.$ This completes the proof of the lemma.

 \end{proof}

In the next lemma, we construct a set of $r$-tuples $\mathcal{A}_r = (A_1, \dots, A_r) $, where each $A_i$ consists of two nonzero elements of the field $\Fq$, such that the condition in Lemma \ref{lemm:Disjoint} is satisfied. Recall that $q\geq 3$.

\begin{lemma}
\label{lem:Gqr}
    There is a set $\mathbb{G}_q^{(r)}$ of exactly $ \left\lfloor \frac{q-1}{2} \right\rfloor^r$ many $r$-tuples $\mathcal{A}_r = (A_1, \dots, A_r) $, where each $A_i$ consists of two nonzero elements of the field $\Fq$ satisfying
$$
\mathcal{A}_r,\; \mathcal{B}_r\in \mathbb{G}_q^{(r)}\implies A_i\cap B_i=\emptyset\text{ for some }1\leq i\leq r.
$$

\end{lemma}

\begin{proof}  
    It is well known that the multiplication group $\mathbb{F}_q^*$ is a cyclic group. Let  $\alpha$ be a generator of this group. Consider the set
    
     $$
     \Delta:= \left\{ \{\alpha^{2i-1},\alpha^{2i}\}: i=1,2,\dots, \left\lfloor \frac{q-1}{2} \right\rfloor \right\}.
     $$
    Clearly $|\Delta|=\lfloor(q-1)/2\rfloor$.  Define
    $$
   \mathbb{G}_q^{(r)}:=\{ \mathcal{A}_r=(A_1,A_2,\ldots,A_r) : A_i\in  \Delta \}.
    $$

If $\mathcal{A}_r, \; \mathcal{B}_r\in \mathbb{G}_q^{(r)}$ then $A_i\cap B_i=\emptyset$ for some $1\leq i\leq r$. If this is not true, then $A_i\cap B_i\neq\emptyset$ for all $i$. But this implies that $A_i= B_i$ as $A_i,\; B_i\in \Delta$ and $\Delta$ is a collection of disjoint subsets of $\Fq^*$. Thus the set $\mathbb{G}_q^{(r)}$ is a set with desired property. Clearly, $| \mathbb{G}_q^{(r)}| = \left\lfloor \frac{q-1}{2} \right\rfloor^r$.

\end{proof}

Now we are ready to prove the main result of this article. The next theorem guarantees the existence of plenty of parity checks orthogonal to $\clm$ on some fixed coordinate to perform the majority voting decoder.

\begin{theorem}
\label{thm:Majority_logic_dec}

Let $\l, \l', m$ be positive integers satisfying $\l \leq \l'$, $m=\l+\l'$ and let $\clm$ be the corresponding affine Grassmann code. There exist a  set $\mathcal{J}$  with  $|\mathcal{J}|=J$ parity checks of $\clm$ that are orthogonal on coordinate corresponding to the $0$ matrix, where
 \begin{equation}
     \label{eq:setJ}
      J:= \sum_{r=1}^{\l}  \left\lfloor \frac{q-1}{2} \right\rfloor^r \frac{\prod\limits_{i=0}^{r-1}(q^{\l}-q^i) \prod\limits_{j=0}^{r-1}(q^{\l'}-q^j)}{ (q-1)^r \prod\limits_{i=0}^{r-1} (q^r-q^i)}.
 \end{equation}

\end{theorem}
\begin{proof}
For every $r$ satisfying $1\leq r\leq \ell$, we can construct, as in Lemma \ref{lem:Gqr}, sets $ \mathbb{G}_q^{(r)}$ of $r$-tuples of two subsets of $\Fq^*$ with  $|\mathbb{G}_q^{(r)}|= \lfloor (q-1)/2\rfloor^r.$ It has been proved in Corollary \ref{cor: ParityCheks} that for any subset $\mathcal{A}_r\in \mathbb{G}_q^{(r)}$ and ordered bases $B=\{x_1,\dots, x_r\}$ and $B^\prime=\{y_1,\dots, y_r\}$ of susbspaces of $r$-dimensional subspaces $V_r$ of $\Fq^\ell$ and $W_r$ of $\Fq^{\ell^\prime}$ respectively, there exist a parity check $w$ of $\clm$ such that 
$$
0\in \Supp (w)\subseteq M(B, B^\prime;\mathcal{A}_r)\cup\{0\}.
$$
Therefore, counting orthogonal parity check at coordinate $0$ matrix is equivalent to counting disjoint sets $M(B, B^\prime;\mathcal{A}_r)$. suppose that $V_r$ and $V_r^\prime$ are two $r$-dimensional subspaces of $\Fq^\ell$ with ordered bases $\{x_1, x_2, \ldots, x_r\}$ and $\{x_1', x_2', \ldots, x_r'\}$ respectively. Similarly, suppose that $W_r$ and $W_r^\prime$ are two $r$-dimensional subspaces of $\Fq^{\ell^\prime}$ with ordered bases $\{y_1, x_2, \ldots, y_r\}$ and $\{y_1', x_2', \ldots, y_r'\}$ respectively. If $(V_r, W_r)\neq (V_r^\prime, W_r^\prime)$ then the sets 
$$
M_r = \left\{ \sum_{i=1}^r a_i x_i^T y_i : a_i \in \mathbb{F}_q^* \right\} \quad \text{and} \quad M_r' = \left\{ \sum_{i=1}^r b_i x_i'^T y_i' : b_i \in \mathbb{F}_q^* \right\}
$$
of rank $r$ matrices are disjoint. This follows simply because the column span of each of matrix of the set $M_r$ is $V_r$ and the row span is $W_r$. Clearly, the number of ordered pair $(V_r, W_r)$ of $r$-dimensional subspaces of $V_r$ of $\Fq^\ell$ and $W_r$ of $\Fq^{\ell^\prime}$ is 
$$
\prod\limits_{i=0}^{r-1}\frac{(q^\l - q^i)(q^{\l'} - q^i)}{(q^r-q^i)^2}.
$$

Now let $(V_r, W_r)$ be an ordered pair of subspaces as above. Let $B=\{x_1,\dots, x_r\}$ be an ordered basis of $V_r$ and let $B^\prime=\{y_1,\dots, y_r\}$ be an ordered basis of $W_r$. Then it follows from  Lemma \ref{lem:Gqr} that the set 
$$
\{M(B, B^\prime; \mathcal{A}_r): \mathcal{A}_r\in \mathbb{G}_q^{(r)}\}
$$
contains subsets of rank $r$-matrices that are disjoint, i.e. $M(B, B^\prime; \mathcal{A}_r)\cap M(B, B^\prime; \mathcal{B}_r)=\emptyset$ for all $\mathcal{A}_r, \; \mathcal{B}_r\in \mathbb{G}_q^{(r)}$, $\mathcal{A}_r\neq \mathcal{B}_r$. Further, the number of ordered bases $B^\prime$ of $W_r$ such that the $j^{\it th}$ vector of any two of such ordered bases is linearly independent for some $j$ is
$$
\frac{\prod\limits_{i=0}^{r-1}(q^r - q^i)}{(q-1)^r}.
$$
We  have also shown in Lemma \ref{lemm:Disjoint}  that if $B_1,\; B_2$ are such two ordered bases then for any $\mathcal{A}_r\in \mathbb{G}_q^{(r)}$, the set $M(B, B_1;\mathcal{A}_r)$ and $M(B, B_2;\mathcal{A}_r)$ are disjoint. Therefore, for a fixed $1\leq r\leq\ell$, there exist 
$$
\left\lfloor \frac{q-1}{2} \right\rfloor^r \frac{\prod\limits_{i=0}^{r-1}(q^{\l}-q^i) \prod\limits_{j=0}^{r-1}(q^{\l'}-q^j)}{ (q-1)^r \prod\limits_{i=0}^{r-1} (q^r-q^i)}.
$$
set $M(B, B^\prime; \mathcal{A}_r)$ of $2^r$ rank $r$ matrices such that any two of them are disjoint. Moreover, corresponding to each such $M(B, B^\prime; \mathcal{A}_r)$, we get a parity check of $\clm$ that are orthogonal on the zero matrix coordinate. Varying $r$ from $1$ to $\ell$ and adding them we get the set $\mathcal{J}$ of $J$ orthogonal parity checks as desired in the theorem.

\end{proof}

The next corollary is the final result of the article. 
\begin{corollary}
  Using majority voting, one can correct $\lfloor J/2\rfloor$ errors for the affine Grassmann code  $\clm$, where $J$ is given by equation \eqref{eq:setJ}. 
\end{corollary}
\begin{proof}
  From Theorem \ref{thm:Majority_logic_dec}, we get a set $\mathcal{J}$ of $J$ parity checks, where $J$ is given by equation \eqref{eq:setJ}, which are orthogonal on the coordinate $0$-matrix. Further, we know from Lemma \ref{lemma: Autgrp} that the automorphism group of $\clm$ acts transitively on the coordinates of codewords of $\clm$. Using this, a set of orthogonal parity checks for each coordinate can be produced. Finally, using Theorem \ref{thm:Majority_logic_dec} we can correct $\lfloor J/2\rfloor$ errors for $\clm$.
\end{proof}

\begin{remark}
\label{counting}

If $q\geq 3$, we see from the proof of theorem \ref{thm:Majority_logic_dec}, using the sets $M(B, B^\prime;\mathcal{A}_r)$ we can produced 
$$
2^r \left\lfloor \frac{q-1}{2} \right\rfloor^r \frac{\prod\limits_{i=0}^{r-1}(q^{\l}-q^i) \prod\limits_{j=0}^{r-1}(q^{\l'}-q^j)}{ (q-1)^r \prod\limits_{i=0}^{r-1} (q^r-q^i)}
$$ 
distinct $\l \times \l'$ matrices of rank $r$ over finite field $\Fq$. If $q$ is odd, then this number is the same as the total number of $\l\times \l'$ matrices of rank $r$ over the finite field $\F_q$. Therefore, the proposed set of parity checks is optimal, and this set can not be enlarged to improve the decoding capability of the code. 
\end{remark}

\begin{remark}
   Let $c \in \clm$ be the transmitted code, and let $w$ be the received codeword with Hamming distance $d_H(c,w) \leq J/2$, where $J$ is as in equation \eqref{eq:setJ}. From theorem \ref{thm:Majority_logic_dec}, we can uniquely recover the codeword $c$. If $\l$ and $m$ are fixed, and $d$ is the minimum distance of the code $\clm$, then  $J/d \longrightarrow 1/2^\l$ if $q\longrightarrow \infty$. Hence, asymptotically we can correct up to $d/2^{\l+1}$ many errors using majority voting decoder proposed in Theorem \ref{thm:Majority_logic_dec}
\end{remark}

\subsubsection*{Acknowledgements} During the work of this article, Prasant Singh was supported by SERB SRG grant SRG/2022/001643 from DST, Govt of India. A significant part of this work was carried out during a visit of Fernando Pi\~nero to IIT-Jammu. The authors wish to thank Sudhir Ghorpade, Mrinmoy Datta, the organizers of the ICAGCTC-2023 conference, and the Mathematics Departments of IIT-Jammu and IIT-Bombay for their generous assistance in facilitating that visit.

\newpage

\bibliographystyle{plain}

\bibliography{Version_2.bib}

\begin{bibdiv}
\begin{biblist}

\bib{BGH2010}{article}{
      author={Beelen, Peter},
      author={Ghorpade, Sudhir~R.},
      author={H{\o}holdt, Tom},
       title={Affine {G}rassmann codes},
        date={2010},
        ISSN={0018-9448,1557-9654},
     journal={IEEE Trans. Inform. Theory},
      volume={56},
      number={7},
       pages={3166\ndash 3176},
         url={https://doi.org/10.1109/TIT.2010.2048470},
      review={\MR{2798982}},
}

\bib{BGT2012}{article}{
      author={Beelen, Peter},
      author={Ghorpade, Sudhir~R.},
      author={H{\o}holdt, Tom},
       title={Duals of affine {G}rassmann codes and their relatives},
        date={2012},
        ISSN={0018-9448,1557-9654},
     journal={IEEE Trans. Inform. Theory},
      volume={58},
      number={6},
       pages={3843\ndash 3855},
         url={https://doi.org/10.1109/TIT.2012.2187171},
      review={\MR{2924405}},
}

\bib{BP2016}{article}{
      author={Beelen, Peter},
      author={Pi\~nero, Fernando},
       title={The structure of dual {G}rassmann codes},
        date={2016},
        ISSN={0925-1022,1573-7586},
     journal={Des. Codes Cryptogr.},
      volume={79},
      number={3},
       pages={451\ndash 470},
         url={https://doi.org/10.1007/s10623-015-0085-3},
      review={\MR{3489752}},
}

\bib{BS2019}{article}{
      author={Beelen, Peter},
      author={Singh, Prasant},
       title={Linear codes associated to skew-symmetric determinantal varieties},
        date={2019},
        ISSN={1071-5797,1090-2465},
     journal={Finite Fields Appl.},
      volume={58},
       pages={32\ndash 45},
         url={https://doi.org/10.1016/j.ffa.2019.03.004},
      review={\MR{3928684}},
}

\bib{BPP2021}{article}{
      author={Beelen, Peter},
      author={Singh, Prasant},
       title={Point-line incidence on {G}rassmannians and majority logic decoding of {G}rassmann codes},
        date={2021},
        ISSN={1071-5797,1090-2465},
     journal={Finite Fields Appl.},
      volume={73},
       pages={Paper No. 101843, 24},
         url={https://doi.org/10.1016/j.ffa.2021.101843},
      review={\MR{4241603}},
}

\bib{DG2015}{incollection}{
      author={Datta, Mrinmoy},
      author={Ghorpade, Sudhir~R.},
       title={Higher weights of affine {G}rassmann codes and their duals},
        date={2015},
   booktitle={Algorithmic arithmetic, geometry, and coding theory},
      series={Contemp. Math.},
      volume={637},
   publisher={Amer. Math. Soc., Providence, RI},
       pages={79\ndash 91},
         url={https://doi.org/10.1090/conm/637/12750},
      review={\MR{3364443}},
}

\bib{GK2013}{article}{
      author={Ghorpade, Sudhir~R.},
      author={Kaipa, Krishna~V.},
       title={Automorphism groups of {G}rassmann codes},
        date={2013},
        ISSN={1071-5797,1090-2465},
     journal={Finite Fields Appl.},
      volume={23},
       pages={80\ndash 102},
         url={https://doi.org/10.1016/j.ffa.2013.04.005},
      review={\MR{3061086}},
}

\bib{M1963}{book}{
      author={Massey, James~L.},
       title={Threshold decoding},
   publisher={Massachusetts Institute of Technology, Research Laboratory of Electronics, Cambridge, MA},
        date={1963},
        note={Tech. Rep. 410},
      review={\MR{154763}},
}

\bib{MS1977}{book}{
      author={MacWilliams, F.~J.},
      author={Sloane, N. J.~A.},
       title={The theory of error-correcting codes. {I}},
      series={North-Holland Mathematical Library},
   publisher={North-Holland Publishing Co., Amsterdam-New York-Oxford},
        date={1977},
      volume={Vol. 16},
        ISBN={0-444-85009-0},
      review={\MR{465509}},
}

\bib{Nogin1996}{incollection}{
      author={Nogin, D.~Yu.},
       title={Codes associated to {G}rassmannians},
        date={1996},
   booktitle={Arithmetic, geometry and coding theory ({L}uminy, 1993)},
   publisher={de Gruyter, Berlin},
       pages={145\ndash 154},
      review={\MR{1394931}},
}

\bib{PS2019}{article}{
      author={Pi\~nero, Fernando},
      author={Singh, Prasant},
       title={The weight spectrum of certain affine {G}rassmann codes},
        date={2019},
        ISSN={0925-1022,1573-7586},
     journal={Des. Codes Cryptogr.},
      volume={87},
      number={4},
       pages={817\ndash 830},
         url={https://doi.org/10.1007/s10623-018-0567-1},
      review={\MR{3923598}},
}

\bib{Reed1954}{article}{
      author={Reed, Irving~S.},
       title={A class of multiple-error-correcting codes and the decoding scheme},
        date={1954},
     journal={Trans. IRE},
      number={PGIT-},
       pages={38\ndash 49},
      review={\MR{89789}},
}

\bib{S2022}{article}{
      author={Singh, Prasant},
       title={Majority logic decoding for certain {S}chubert codes using lines in {S}chubert varieties},
        date={2022},
        ISSN={0018-9448,1557-9654},
     journal={IEEE Trans. Inform. Theory},
      volume={68},
      number={2},
       pages={795\ndash 805},
         url={https://doi.org/10.1109/TIT.2021.3124549},
      review={\MR{4395442}},
}

\bib{TVN2007}{book}{
      author={Tsfasman, Michael},
      author={Vl\u{a}du\c{t}, Serge},
      author={Nogin, Dmitry},
       title={Algebraic geometric codes: basic notions},
      series={Mathematical Surveys and Monographs},
   publisher={American Mathematical Society, Providence, RI},
        date={2007},
      volume={139},
        ISBN={978-0-8218-4306-2},
         url={https://doi.org/10.1090/surv/139},
      review={\MR{2339649}},
}

\end{biblist}
\end{bibdiv}
  
\end{document}